%% file: partition.tex
\documentclass[a4paper]{article}
\usepackage{refcount}
\usepackage{graphicx}
\usepackage{authblk}

\usepackage{amsmath,amssymb,amsthm}
\newtheorem{theorem}{Theorem}[section]
\newtheorem{lemma}[theorem]{Lemma}
\newtheorem{proposition}[theorem]{Proposition}
\newtheorem{claim}[theorem]{Claim}
\usepackage[linesnumbered,ruled,vlined]{algorithm2e}

\newcommand{\rank}{\mathop{\rm rank}}

\newcommand{\argmin}{\mathop{\rm argmin}}
\newcommand{\innerop}{\mathop{\rm Op^{(2)}}}
\newcommand{\outerop}{\mathop{\rm Op^{(1)}}}
\newcommand{\ALG}{{\rm ALG}}
\newcommand{\OPT}{{\rm OPT}}
\newcommand{\texorpdfstring}[2]{#1}

\newcommand{\T}{{\rm T}}
\newcommand{\F}{{\rm F}}
\renewcommand{\mid}{\,:\,}

\usepackage{version}
\excludeversion{conference}
\includeversion{fullpaper}

\usepackage{CJKutf8}
\usepackage{ifthen}
\newcommand{\COMM}[2]{{
\begin{CJK}{UTF8}{ipxm}
\ifthenelse{\equal{#1}{YK}}{\color{blue}}{
\ifthenelse{\equal{#1}{HS}}{\color{red}}{
\ifthenelse{\equal{#1}{KK}}{\color{green}}}}
[#1: #2]
\end{CJK}
}}

\makeatletter
\def\iddots{\mathinner{\mkern1mu\raise\p@
    \hbox{.}\mkern2mu\raise4\p@\hbox{.}\mkern2mu
    \raise7\p@\vbox{\kern7\p@\hbox{.}}\mkern1mu}}
\makeatother

\usepackage[top=30truemm,bottom=30truemm,left=25truemm,right=25truemm]{geometry}

\title{Optimal Matroid Partitioning Problems}
\date{}

\author[1]{Yasushi Kawase}
\author[2]{Kei Kimura}
\author[3]{Kazuhisa Makino}
\author[4]{Hanna Sumita}
\affil[1]{Tokyo Institute of Technology, Japan. \texttt{kawase.y.ab@m.titech.ac.jp}}
\affil[2]{Toyohashi University of Technology, Japan. \texttt{kimura@cs.tut.ac.jp}}
\affil[3]{Kyoto University, Japan. \texttt{makino@kurims.kyoto-u.ac.jp}}
\affil[4]{National Institute of Technology, Japan.
 \texttt{sumita@nii.ac.jp}}

\begin{document}

\maketitle
\input{abstract}

\input{introduction}
\input{preliminaries}
\input{msM}
\input{tractable}

\input{intractable}

\section*{Acknowledgments}
The first author is supported by JSPS KAKENHI Grant Number JP16K16005.
The second author is supported by JSPS KAKENHI Grant Number JP15H06286.
The third author is supported by JSPS KAKENHI Grant Number JP24106002, JP25280004, JP26280001, and JST CREST Grant Number JPMJCR1402, Japan.
The forth author is supported by JST ERATO Grant Number JPMJER1201, Japan.

\bibliography{partition}

\begin{conference}
\clearpage
\appendix
\input{src/appendix}
\end{conference}

\end{document}

%% file: abstract.tex
\begin{abstract}
This paper studies optimal matroid partitioning problems for various objective functions.
In the problem, we are given a finite set $E$ and $k$ weighted matroids $(E, \mathcal{I}_i, w_i)$, $i = 1, \dots, k$, 
and our task is to find a minimum partition $(I_1,\dots,I_k)$ of $E$ such that $I_i \in \mathcal{I}_i$ for all $i$.
For each objective function, we give a polynomial-time algorithm or prove NP-hardness.
In particular, for the case when the given weighted matroids are identical and 
the objective function is the sum of the maximum weight in each set (i.e., $\sum_{i=1}^k\max_{e\in I_i}w_i(e)$), we show that the problem is strongly NP-hard but admits a PTAS.
\end{abstract}

%% file: introduction.tex
\section{Introduction}
The \emph{matroid partitioning problem} is one of the most fundamental problems in combinatorial optimization. 
In this problem, we are given a finite set $E$ and $k$ matroids $(E, \mathcal{I}_i)$, $i = 1, \dots, k$, and our task is to find a partition $(I_1,\dots,I_k)$ of $E$ such that $I_i \in \mathcal{I}_i$ for all $i$. 
We say that such a partition $(I_1, \ldots, I_k)$ of $E$ is \emph{feasible}. 
The matroid partitioning problem has been eagerly studied in a series of papers investigating structures of matroids. 
See e.g., \cite{Edm65,Edmonds1970,EdF65,KorteVygen2002,Schrijver2003} for details. 
In this paper, we study weighted versions of the matroid partitioning problem. 
Namely, we assume that each matroid $(E, \mathcal{I}_i)$ has a weight function $w_i:E \to \mathbb{R}_+$. 
We consider several possible objective functions of the matroid partitioning problem.

Let $\outerop$ and $\innerop$ denote two mathematical operators taken from $\{\max, \min, \sum\}$. 
For any partition $P=(I_1, \dots , I_k)$ of $E$,
we call $\outerop_{i=1,\dots ,k}\innerop_{e \in I_i} w_i(e)$ the \emph{$(\outerop,\innerop)$-value} of $P$.
For example, $(\sum,\min)$-value of $P$ denotes $\sum_{i=1,\dots ,k} \min_{e\in I_i} w_i(e)$.

We define the \emph{minimum $(\outerop,\innerop)$-value matroid partitioning problem} as the one for finding a feasible partition with minimum $(\outerop,\innerop)$-value.
The maximum problems are defined analogously. 
These matroid partitioning problems are natural to study, and have many applications in various areas such as scheduling and combinatorial optimization. 
We note that all the matroids and/or all the weights may be identical in case such as scheduling with identical machines.

The minimum $(\sum,\sum)$-value matroid partitioning problem is reducible to the \emph{weighted matroid intersection problem}, and vice versa~\cite{Edmonds1970}.
Here, the weighted matroid intersection problem is to find a maximum weight subset that is simultaneously independent in two given matroids. 
It is known that this problem is polynomially solvable, and many papers have worked on algorithmic aspects of this problem~\cite{KorteVygen2002,Schrijver2003}. 
Generalizations of the weighted matroid intersection problem have also been studied~\cite{Lawler1971,Lovasz1980,LeeSV2013}.

Special cases of the minimum $(\max,\sum)$-value matroid partitioning problem have been extensively addressed in the scheduling literature under the name of the minimum makespan scheduling.
Since this problem is NP-hard, many papers have proposed polynomial-time approximation algorithms. 
We remark that most papers focused on subclasses of matroids as inputs: 
for example, free matroids~\cite{LST90,VeW14}, partition matroids~\cite{WuY07,WuY08,LiL14}, uniform matroid~\cite{KeW93,BKK98,Dem01}, and general matroids~\cite{WuY08}.
Approximation algorithms for the maximum $(\min,\sum)$-value matroid partitioning problem are also well-studied, see, e.g., \cite{CHL02,HTZ03,WuY07,LiL14}.

The other matroid partitioning problems also have many applications, 
and yet they are not much studied especially for general matroids.
We here describe some examples of applications.
\begin{description}
\begin{fullpaper}
\item[Minimum bottleneck spanning tree] 
  Consider the minimum bottleneck spanning tree problem:
  given a connected undirected weighted graph $G=(V,E;w)$,
  we are asked to find a spanning tree whose maximum weight edge is minimized.
  Let $(E,\mathcal{I}_1)$ be a graphic matroid obtained from $G$ and let $(E,\mathcal{I}_2)$ be its dual.
  Then, this problem is formulated as the minimum $(\max,\max)$-value matroid partitioning problem with $(E,\mathcal{I}_1)$, $(E,\mathcal{I}_2)$, $w_1 = w$ and $w_2 = \mathbf{0}$, where the $(\max,\max)$-value is $\max_{i=1,2} \max_{e \in I_i} w_i(e)$. 
\end{fullpaper}
\item[Maximum total capacity spanning tree partition] 
  Assume that we are given an undirected weighted graph $G=(V,E;w)$,
  which can be partitioned into $k$ edge-disjoint spanning trees.
  The maximum total capacity spanning tree partition problem
  is to compute a partition of the edges into $k$ edge-disjoint spanning trees
  such that the total of the minimum weight in each spanning tree is maximized.
  Then, the problem can be written as the maximum $(\sum,\min)$-value matroid partitioning problem having $k$ identical graphic matroids, where the $(\sum,\min)$-value is $\sum_{i=1}^k \min_{e \in I_i} w(e)$. 
\item[Minimum total memory of a scheduling] 
  In this problem we are also given $n$ jobs $E$ and $k$ identical machines,
  and each job needs to be scheduled on exactly one machine.
  In addition, we are given size $s(e)$ of job $e\in E$.
  The set of feasible allocation for each machine $i$ is represented by a family of independent sets $\mathcal{I}_i$ of a matroid.
  The goal of the problem is to minimize the total memory needed, i.e., $(\sum,\max)$-value $\sum_{i=1}^k \max_{e \in I_i} s(e)$. 
\end{description}
Burkard and Yao~\cite{BuY90} showed that the minimum $(\sum, \max)$-value matroid problem can be solved by a greedy algorithm for a subclass of matroids, which includes partition matroids. 
Dell'Olmo et al.~\cite{DHP05} investigated optimal matroid partitioning problems where the input matroids are identical partition matroids.

The goal of our paper is to analyze the computational complexity of these matroid partitioning problems for general matroids.




\subsection*{Our results}
We first show that the maximization problems can be reduced to the minimization problems. 
For example, the maximum $(\sum, \min)$-value matroid partitioning problem can be transformed to the minimum $(\sum, \max)$-value matroid partitioning problem. 
Hence, we focus only on the minimization problems.

Our main result is to analyze the computational complexity of the minimum $(\sum, \max)$-value matroid partitioning problem. 
This problem contains the maximum total capacity spanning tree partitioning problem and the minimum total memory scheduling problem. 
We first show that the problem is strongly NP-hard even when the matroids and weights are respectively identical. 
However, for such instances, we also propose a \emph{polynomial-time approximation scheme} (PTAS), i.e., a polynomial-time algorithm that outputs a $(1+\varepsilon)$-approximate solution for each fixed $\varepsilon > 0$. 
Our PTAS computes an approximate solution by two steps: guess the maximum weight in each $I_i^*$ for an optimal solution $(I_1^*,\dots,I_k^*)$, and check the existence of such a feasible partition. 
We remark that the number of possible combinations of maximum weights is $|E|^k$ and it may be too large. 
To reduce the possibility, we use rounding techniques in the design of the PTAS.  
First, we guess the maximum weight in $I_i^*$ for only $s$ indices.
Furthermore, we round the weight of each element and reduce the number of different weights to a small number $r$. 
Then, now we have $r^s$ possibilities. 
To obtain the approximation ratio $(1+\varepsilon)$, we need to set $r$ and $s$ to be $\Omega(\log k)$ respectively, and hence the number of possibilities $r^s$ is still large.
Our idea to tackle this is to enumerate \emph{sequences} of maximum weights in the nonincreasing order. 
This enables us to reduce the number of possibilities to $\binom{r+s-1}{r}~(\le 2^{r+s-1})$.
This implies that our algorithm is a PTAS.

Moreover, for the $(\sum, \max)$ case with general inputs, we provide an $\varepsilon k$-approximation algorithm for any $\varepsilon > 0$. 
The construction is similar to the identical case. 
We also prove the NP-hardness even to approximate the problem within a factor of $o(\log k)$.

For the $(\min,\min)$, $(\max,\max)$, $(\min,\max)$, and $(\min,\sum)$ cases, we provide polynomial-time algorithms. 
The main idea of these algorithms is a reduction to the feasibility problem of the matroid partitioning problem. 
For the $(\max,\min)$ and $(\sum,\min)$ cases, we give polynomial-time algorithms when the matroids and weights are respectively identical, and prove strong NP-hardness even to approximate for the general case.
These results are summarized in Table~\ref{table:results} with their references.

\begin{table}[htb]
\caption{The time complexity of the optimal matroid partitioning problems (the results of the paper are in bold).
Identical case means $\mathcal{I}_1 = \dots =\mathcal{I}_k$ and $w_1 = \dots = w_k$.}\label{table:results}
\centering
\begin{tabular}{clll}
    \hline
    objective     &identical case   &general case                                 &reference\\ \hline
    $(\sum,\sum)$ &P                &P                                            &\cite{Edmonds1970,Frank1981}\\
    $(\max,\sum)$ &SNP-hard         &SNP-hard                                     &\cite{garey1979cai}\\
    $(\sum,\max)$ &\textbf{PTAS}    &\textbf{$\varepsilon k$-approx. }               &Section \ref{sec:msM}\\
                  &\textbf{SNP-hard}&\textbf{NP-hard even for $o(\log k)$-approx.}&\\
    $(\min,\min)$ &\textbf{P}       &\textbf{P}                                   &Section \ref{sec:tractable}\\
    $(\max,\max)$ &\textbf{P}       &\textbf{P}                                   &Section \ref{sec:tractable}\\
    $(\min,\max)$ &\textbf{P}       &\textbf{P}                                   &Section \ref{sec:tractable}\\
    $(\min,\sum)$ &\textbf{P}       &\textbf{P}                                   &Section \ref{sec:tractable}\\
    $(\max,\min)$ &\textbf{P}       &\textbf{NP-hard even to approximate}         &Sections \ref{sec:tractable}, \ref{sec:intractable}\\
    $(\sum,\min)$ &\textbf{P}       &\textbf{NP-hard even to approximate}         &Sections \ref{sec:tractable}, \ref{sec:intractable}\\
    \hline
\end{tabular}
\end{table}

\begin{fullpaper}
\subsection*{The organization of the paper}
In Section \ref{sec:preliminaries}, we define optimal matroid partitioning problems and show basic properties of theses problems.
In particular, we see that the maximization problems are reducible to minimization problems.
In Section \ref{sec:msM}, we prove that the minimum $(\sum, \max)$-value matroid partitioning problem with identical matroids and identical weights is strongly NP-hard but admits a PTAS.
We also show that the problem for general case has no $o(\log k)$-approximation algorithm but has an $\varepsilon k$-approximation algorithm. 
Section~\ref{sec:tractable} deals with tractable cases among other problems, and Section \ref{sec:intractable} gives hardness results.
\end{fullpaper}

%% file: preliminaries.tex
\section{Preliminaries}\label{sec:preliminaries}
A \emph{matroid} is a set system $(E,\mathcal{I})$ with the following properties:
(I1) $\emptyset\in\mathcal{I}$,
(I2) $X\subseteq Y\in\mathcal{I}$ implies $X\in\mathcal{I}$, and
(I3) $X,Y\in\mathcal{I}$, $|X|<|Y|$ implies the existence of $e\in Y\setminus X$ such that $X\cup\{e\}\in\mathcal{I}$.
A set $I\subseteq \mathcal{I}$ is said to be \emph{independent}, and
an inclusion-wise maximal independent set is called a \emph{base}.
We denote the set of bases of $(E,\mathcal{I})$ by $B(\mathcal{I})$.
All bases of a matroid have the same cardinality, which is called the \emph{rank} of the matroid and is denoted by $\rank(\mathcal{I})$.
For any $B_1,B_2\in B(\mathcal{I})$ and $e_1\in B_1\setminus B_2$,
there exists $e_2\in B_2\setminus B_1$ such that
$B_1-e_1+e_2\in B(\mathcal{I})$ and $B_2-e_2+e_1\in B(\mathcal{I})$.

For a matroid $(E,\mathcal{I})$, a subset $A\subseteq E$, and a nonnegative integer $l\in\mathbb{Z}_{+}$,
define
\begin{conference}
$\mathcal{I}|A=\{X \mid A\supseteq X\in \mathcal{I}\}$,
$\mathcal{I}\setminus A=\{X\setminus A\mid X\in\mathcal{I}\}$,
$\mathcal{I}/A=\{X\subseteq E\setminus A\mid \rank(X\cup A)-\rank(A)=|X|\}$, and
$\mathcal{I}^{(l)}=\{X\in\mathcal{I}\mid |X|\le l\}$.
\end{conference}
\begin{fullpaper}
\begin{align*}
  \mathcal{I}|A&=\{X \mid A\supseteq X\in \mathcal{I}\},&
  \mathcal{I}\setminus A=\{X\setminus A\mid X\in\mathcal{I}\},\\
  \mathcal{I}/A&=\{X\subseteq E\setminus A\mid \rank(X\cup A)-\rank(A)=|X|\}, &
  \mathcal{I}^{(l)}=\{X\in\mathcal{I}\mid |X|\le l\}.
\end{align*}
\end{fullpaper}
We call $(A,\mathcal{I}|A)$, $(E\setminus A,\mathcal{I}\setminus A)$, $(E\setminus A,\mathcal{I}/A)$, and $(E,\mathcal{I}^{(l)})$, respectively, the \emph{restriction}, \emph{deletion}, \emph{contraction}, and \emph{truncation} of $(E,\mathcal{I})$.
It is well known that $(A,\mathcal{I}|A)$, $(A,\mathcal{I}\setminus A)$, $(E\setminus A,\mathcal{I}/A)$, and $(E,\mathcal{I}^{(l)})$
are all matroids.
Given matroids $\mathcal{M}_1 = (E_1,\mathcal{I}_1)$ and $\mathcal{M}_2 = (E_2,\mathcal{I}_2)$,
we define the \emph{matroid union}, denoted by $\mathcal{M}_1 \vee \mathcal{M}_2$, to be $(E_1\cup E_2,\mathcal{I}_1\vee\mathcal{I}_2)$
where $\mathcal{I}_1\vee\mathcal{I}_2=\{I_1\cup I_2\mid I_1\in\mathcal{I}_1,~I_2\in\mathcal{I}_2\}$. 
Any matroid union is also a matroid.
For more details of matroids, see e.g.,~\cite{oxley1992mt}.

\subsection{Model}
Throughout the paper, we assume that every matroid is given by an independence oracle, which checks whether a given set is independent.
Let $k$ be a positive integer. 
We denote $[k]=\{1,\dots,k\}$. 
Let $(E,\mathcal{I}_i)$ be a matroid and $w_i:E\to\mathbb{R}_+$ be a nonnegative weight function for $i\in [k]$. 
We denote $n=|E|$. 
For any $k$ sets $I_1,\dots,I_k \subseteq E$, we call $(I_1,\dots,I_k)$ a \emph{feasible partition} of $E$ if it satisfies that $\bigcup_{i\in[k]} I_i=E$, $I_i\ne\emptyset~(\forall i\in[k])$\footnotemark, $I_i\cap I_j=\emptyset~(\forall i,j\in[k]$,~$i\ne j)$, and $I_i\in\mathcal{I}_i \ (\forall i \in [k])$.
In particular, $(I_1,\dots,I_k)$ is said to be a \emph{base partition} if it is a feasible partition and $I_i\in B(\mathcal{I}_i)$ for all $i\in[k]$.
\footnotetext{We remark that the condition $I_i\ne\emptyset~(\forall i\in[k])$ is imposed to make the objective function well-defined.
  Moreover, if we define
  $\max_{e\in\emptyset}w_i(e)=0$, 
  $\min_{e\in\emptyset}w_i(e)=\infty$, 
  and $\sum_{e\in\emptyset}w_i(e)=0$, 
  then we can reduce the problem where empty sets are allowed to our problem by adding dummy elements.
}
For two operators $\outerop\in\{\max,\min,\sum\}$ and $\innerop\in\{\max,\min,\sum\}$, 
we define the $(\outerop,\innerop)$-value of a feasible partition $(I_1,\dots,I_k)$ as
\begin{conference}
\(\outerop_{i\in[k]} \innerop_{e\in I_i} w_i(e).\)
\end{conference}
\begin{fullpaper}
\begin{align*}
  \outerop_{i\in[k]} \innerop_{e\in I_i} w_i(e).
\end{align*}
\end{fullpaper}
In this article, we study the following minimization problem:
\begin{conference}
\begin{align*}
\min\nolimits_{(I_1,\dots,I_k):~\text{feasible partition}}\outerop\nolimits_{i\in[k]} \innerop\nolimits_{e\in I_i} w_i(e).
\end{align*}
\end{conference}
\begin{fullpaper}
\begin{align*}
\min_{(I_1,\dots,I_k):~\text{feasible partition}}\outerop_{i\in[k]} \innerop_{e\in I_i} w_i(e).
\end{align*}
\end{fullpaper}
We refer to the problem as the \emph{minimum $(\outerop,\innerop)$-value matroid partitioning problem}.
We write a problem instance as $(E,(\mathcal{I}_i,w_i)_{i\in[k]})$.
If $(\mathcal{I}_i,w_i)$ are identical for all $i\in[k]$, we write $(E,(\mathcal{I},w),k)$.
For the identical case, we can consider the partitioning problem where $k$ is also a variable. 
This problem can be solved by solving $(E,(\mathcal{I},w),i)$ for $i=1, \ldots, n$. 
Thus it suffices to focus on the problem where $k$ is given. 

It is known to be easy to decide whether there exists a feasible partition or not.
Moreover, the minimum $(\sum,\sum)$-value matroid partitioning problem can be solved in polynomial time.
These facts are useful to show our results later.
\begin{theorem}[\cite{Edmonds1970,Frank1981}]\label{thm:mss}
There exists a polynomial-time algorithm 
that decides whether or not there exists a feasible partition for any given matroids $(E,\mathcal{I}_1),\dots,(E,\mathcal{I}_k)$. 
Moreover, if it exists, we can find a feasible partition with minimum $(\sum,\sum)$-value in polynomial time.
\end{theorem}

\subsection{Basic properties}\label{sec:basic}
In this subsection, we provide basic properties of the partitioning problems.
These properties imply that the minimization and maximization versions of matroid partitioning problems can be reduced to each other.

We first observe that we only need to consider base partitioning problems.
Let $\mathcal{M}_i = (E,\mathcal{I}_i)$ be a matroid for $i\in [k]$. 
We add dummy elements so that any feasible partition is a base partition. 
To describe this precisely, we denote $r=\sum_{i\in[k]}\rank(\mathcal{I}_i)-|E|$. 
We remark that $r\ge 0$ if $E$ has a feasible partition, since \(|E|=\sum_{i\in[k]}|I_i|\le\sum_{i\in[k]}\rank( \mathcal{I}_i) \) holds for any feasible partition $(I_1,\dots,I_k)$.
Then let $D=\{d_1,\dots,d_r\}$ be a set of dummy elements. 
Note that $E\cap D=\emptyset$. 
We define two matroids $\mathcal{M}'_i=(D,\mathcal{I}'_i)$ and $\overline{\mathcal{M}}_i=(E\cup D,\overline{\mathcal{I}}_i)$ for each $i\in[k]$ by
\begin{conference}
$\mathcal{I}_i'=\{D'\subseteq D\mid |D'|\le \rank(\mathcal{I}_i)-1\}$ and
$\overline{\mathcal{I}}_i=\{I\cup D'\mid I\in\mathcal{I}_i,~D'\in\mathcal{I}_i',~|I\cup D'|\le \rank(\mathcal{I}_i)\}$. 
\end{conference}
\begin{fullpaper}
\begin{align*}
\mathcal{I}_i'&=\{D'\subseteq D\mid |D'|\le \rank(\mathcal{I}_i)-1\},\\
\overline{\mathcal{I}}_i&=\{I\cup D'\mid I\in\mathcal{I}_i,~D'\in\mathcal{I}_i',~|I\cup D'|\le \rank(\mathcal{I}_i)\}. 
\end{align*}
\end{fullpaper}
Namely, $\mathcal{M}_i'$ is a uniform matroid of rank $(\rank(\mathcal{I}_i)-1)$,
and $\overline{\mathcal{M}}_i$ is the $\rank(\mathcal{I}_i)$-truncation of the matroid union $\mathcal{M}_i\vee \mathcal{M}'_i$.
Then, we have the following proposition.
\begin{conference}
The proof can be found in Appendix.
\end{conference}
\begin{proposition}\label{prop:basepartitioning}
For any $(E,(\mathcal{I}_i,w_i)_{i\in[k]})$, its minimum $(\outerop,\innerop)$-value is the same as the minimum $(\outerop,\innerop)$-value for $(E\cup D,(\overline{\mathcal{I}}_i,\overline{w}_i)_{i\in[k]})$, where 
\begin{align*}
  \overline{w}_i(e)=
  \begin{cases}
    w_i(e)             &(e\in E),\\
    \min_{e\in E}w_i(e)&(e\in D,~\innerop=\max),\\
    \max_{e\in E}w_i(e)&(e\in D,~\innerop=\min),\\
    0                  &(e\in D,~\innerop=\sum).
  \end{cases}
\end{align*}
\end{proposition}
\begin{fullpaper}
\begin{proof}
We observe that by the definition of $\overline{w}_i$, we have $\innerop_{e\in I_i} w_i(e)=\innerop_{e\in I_i\cup D_i}\overline{w}_i(e)$ for any $i\in[k]$, $I_i\subseteq E$ and $D_i\subseteq D$ such that $|I_i|\ge 1$.
Suppose that $(I_1,\dots,I_k)$ attains the minimum $(\outerop,\innerop)$-value for $(E,(\mathcal{I}_i,w_i)_{i\in[k]})$ and 
$(\overline{I}_1,\dots,\overline{I}_k)$ attains the minimum $(\outerop,\innerop)$-value for $(E\cup D,(\overline{\mathcal{I}}_i,\overline{w}_i)_{i\in[k]})$. 
Let $I_i'=I_i\cup\{d_k \mid \sum_{j=1}^{i-1}(r_j-|I_j|)<k\le\sum_{j=1}^i(r_j-|I_j|)\}$ 
and $\overline{I}_i'=\overline{I}_i\setminus D$.
Then, since $(\overline{I}_1',\dots,\overline{I}_k')$ is a feasible partition of $E$ with respect to $(E,(\mathcal{I}_i,w_i)_{i\in[k]})$,
we have
\begin{align*}
  \outerop_{i\in[k]} \innerop_{e\in \overline{I}_i} \overline{w}_i(e)
  &=\outerop_{i\in[k]} \innerop_{e\in \overline{I}_i'} w_i(e)
  \ge \outerop_{i\in[k]} \innerop_{e\in I_i} w_i(e).
\end{align*}
On the other hand, 
since $(I_1',\dots,I_k')$ is a feasible partition of $E\cup D$ with respect to $(E\cup D,(\overline{\mathcal{I}}_i,\overline{w}_i)_{i\in[k]})$,
we have
\begin{align*}
  \outerop_{i\in[k]} \innerop_{e\in I_i} w_i(e)
  &=\outerop_{i\in[k]} \innerop_{e\in I_i'} \overline{w}_i(e)
  \ge \outerop_{i\in[k]} \innerop_{e\in \overline{I}_i} \overline{w}_i(e).
\end{align*}
Thus, we obtain $\outerop_{i\in[k]} \innerop_{e\in I_i} w_i(e) = \outerop_{i\in[k]} \innerop_{e\in \overline{I}_i} \overline{w}_i(e)$
and the proposition holds.
\end{proof}
\end{fullpaper}

We remark that the same property holds for the maximization problem.
\begin{conference}
For more details, see Proposition~\ref{prop:maxtomin} in Appendix.
\end{conference}
\begin{fullpaper}
In the following, we assume $|E|=\sum_{i\in[k]}\rank(\mathcal{I}_i)$.
We next show that the maximization problems are reducible to the minimization ones.
\begin{proposition}
  For any feasible partition $(I_1,\dots,I_k)$ for $(E,\mathcal{I}_i)_{i\in [k]}$, it is 
  an optimal solution for the minimum $(\outerop,\innerop)$-value matroid partitioning problem instance $(E,(\mathcal{I}_i,w_i)_{i\in [k]})$
  if and only if
  it is optimal for the maximum $(\widetilde{\outerop},\widetilde{\innerop})$-value matroid partitioning problem instance $(E,(\mathcal{I}_i,w_i')_{i\in [k]})$,
  where $w^{\max}=\max_{i\in[k]}\max_{e\in E}w_i(e)$, 
  \begin{align*}
&\widetilde{\min}=\max, \ \widetilde{\max}=\min, \ \widetilde{\sum}=\sum, \quad \text{and}\\
    &w_i'(e)=\begin{cases}\frac{|E|}{\rank(\mathcal{I}_i)}\cdot w^{\max}-w_i(e)&(\outerop \in\{\min,\max\},\innerop=\sum),\\w^{\max}-w_i(e)&(\text{otherwise}).\end{cases}
  \end{align*}
\end{proposition}
\begin{proof}
By a simple calculation, we have
\begin{align*}
  \outerop_{i\in[k]} \innerop_{e\in I_i} w_i(e)+\mathop{\widetilde{\outerop}}_{i\in[k]}\, \mathop{\widetilde{\innerop}}_{e\in I_i} w_i'(e)&=
  \begin{cases}
    w^{\max}&(\outerop,\innerop \in\{\min,\max\}),\\
    k\cdot w^{\max}&(\outerop=\sum,~\innerop\in\{\min,\max\}),\\
    |E|\cdot w^{\max}&(\innerop=\sum)
  \end{cases}
\end{align*}
for any feasible partition $(I_1,\dots,I_k)$.
Here, the right hand side is a constant, and hence the proposition holds.
\end{proof}
\end{fullpaper}
We note that these reductions above are not approximation factor preserving.
Hence, 
the (in)approximability of the maximization problems are not deduced from that of the minimization problems.



%% file: msM.tex
\section{The minimum \texorpdfstring{$(\sum,\max)$}{(sum,max)}-value matroid partitioning problem}\label{sec:msM}
In this section, we study the minimum $(\sum,\max)$-value matroid partitioning problem.
We first deal with the case where the matroids and weights are respectively identical and then go to the general case.

\subsection{Strong NP-hardness of the identical case}
We first prove that the minimum $(\sum,\max)$-value matroid partitioning problem is strongly NP-hard even if the matroids and weights are respectively identical.

To prove this, we use the \emph{densest $l$-subgraph} problem, which is known to be strongly NP-hard~\cite{FeigePK2001}.
The densest $l$-subgraph problem is, given a graph $G$ and an integer $l$, to find a subgraph of $G$ induced on $l$ vertices that contains the largest number of edges.

\begin{fullpaper}
In our reduction, we use the following property on a partition matroid. 
Let $(E, \mathcal{I})$ be a partition matroid defined by $\mathcal{I} = \{ I \mid |I \cap S_i| \leq \eta_i \ (i\in [p]) \}$, where $(S_1,\dots,S_p)$ is a partition of $E$, and $\eta_1,\dots,\eta_p$ are positive integers. 
In addition, we assume that $|S_i|=\eta_i\cdot k$ for each $i\in [p]$ so that $E$ can be partitioned into $k$ bases of $\mathcal{I}$. 
Then, for any weight $w$, we can construct greedily an optimal partition to the instance $(E, (\mathcal{I}, w), k)$ of the minimum $(\sum, \max)$-value matroid partitioning problem. 
\begin{lemma}[\cite{BuY90}]\label{lem:partition}
Let $(E, \mathcal{I})$ be any partition matroid with $|S_i|=\eta_i\cdot k \ (\forall i \in [p])$, and let $w$ be any weight. 
Let $I_{i,j}$ consist of $\eta_i$ elements with the $\eta_i$ largest weights in $S_i\setminus(\bigcup_{h=1}^{j-1}I_{i,h})$. 
Then $(\bigcup_{i\in [p]} I_{i,1},\dots,\bigcup_{i\in [p]} I_{i,k})$ is an optimal solution to $(E, (\mathcal{I}, w), k)$. 
\end{lemma}
\begin{proof}
Let $(I_1^*,\dots,I_k^*)$ be an optimal partition.
Without loss of generality, we may assume that $\max_{e\in I_1^*}w(e)\ge\dots\ge \max_{e\in I_k^*}w(e)$. 
Let $j$ be any index in $[k]$. 
In addition, let $(i', e_j)$ be the pair of an index and an element attaining $\max_{i \in [p]} \max_{e \in I_{i,j}}w(e)$. 
We claim that $\max_{e \in I^*_j} w(e) \geq w(e_j)$. 
To show this, we suppose the contrary. 
We denote $S = \bigcup_{h < j} I_{i', h} \cup \{e_j\}$. 
Note that $|S| = (j-1)\eta_{i'}+1$ and $w(e) \geq w(e_j)$ for all $e \in S$. 
Since $(E, \mathcal{I})$ is a partition matroid, at most $(j-1)\eta_{i'}$ elements in $S$ are contained in $I^*_1, \ldots, I^*_{j-1}$. 
By assumption $\max_{e \in I^*_j} w(e) < w(e_j)$, there is an index $\ell > j$ such that $I^*_{\ell}$ has some element $e' \in I^*_{\ell} \cap S$. 
Then we have $\max_{e \in I^*_{\ell}}w(e) \geq w(e') \geq w(e_j) > \max_{e\in I^*_j} w(e)$, which contradicts the assumption $\max_{e\in I^*_j} w(e) \geq \max_{e \in I^*_{\ell}}w(e)$. 

Thus, we have
\begin{align*}
\max_{e\in I_j^*}w(e)
\ge w(e_j)=\max_{i\in [p]}\max_{e\in I_{i,j}}w(e)=\max_{e\in \bigcup_{i\in [p]}I_{i,j}}w(e).
\end{align*}
Therefore, $(\bigcup_{i\in [p]} I_{i,1},\dots,\bigcup_{i\in [p]} I_{i,k})$ is also an optimal solution.
\end{proof}
\end{fullpaper}

\begin{theorem}
The minimum $(\sum,\max)$-value matroid partitioning problem is strongly NP-hard even if the matroids and weights are identical.
\end{theorem}
\begin{proof}
Let $G=(V, F)$ be an instance of the densest $l$-subgraph problem. 
We denote $V=\{1,\dots,n\}$, $F=\{f_1,\dots,f_m\}$, and $f_i=\{u_i,v_i\}$.
For any vertex set $T \subseteq V$, we denote $F[T]=\{ \{u,v\} \in F\mid \{u,v\} \subseteq T\}$. 

To solve the densest $l$-subgraph problem, it suffices to find a set of $n-l$ vertices such that the set of the other $l$ vertices attain $\max_{T \subseteq V} |F[T]|$. 
We construct a matroid so that every feasible partition of the ground set corresponds to some set of $n-l$ vertices in $V$, and the $(\sum, \max)$-value is the number of edges in the induced subgraph by the other $l$ vertices. 

Let $V' = \{n+1, \ldots, n+2m\}$ be a set of dummy vertices. 
For each $i \in V \cup V'$, we define a set $E_i$ of $n+2m-1$ elements as
\begin{fullpaper}%
\begin{align*}
E_i = \{ e_{ij}\mid j\in\{1,\dots,n+2m-1\} \}. 
\end{align*}
\end{fullpaper}%
\begin{conference}%
$E_i = \{ e_{ij}\mid j\in\{1,\dots,n+2m-1\}\}$. 
\end{conference}%
Let 
\begin{align*}
  \begin{conference}\textstyle\end{conference}
  E=\bigcup_{i=1}^{n+2m} E_i
  \quad\text{and}\quad
  \mathcal{I}=\{I\subseteq E\mid |I|\le n+2m-1,~|I\cap E_i|\le 1~(\forall i\in [n+2m])\}. 
\end{align*}
The resulting matroid is denoted by $(E, \mathcal{I})$, which is a $(n+2m-1)$-truncation of a partition matroid.
We set $k=n+2m$. 
The weights of elements are defined as follows: 
\begin{itemize}
\item for each $j =1, \ldots, l-1$, set $w(e_{ij})=0 \quad (\forall i\in[n+2m])$; 

\item for each $j = l+2m, \ldots, n+2m-1$, set $w(e_{ij})=m$ if $i \leq n$, and $w(e_{ij})=2m^2$ if $i \geq n+1$;

\item set $w(e_{ij}) \ (j =l, l+1, \ldots, l+2m-1)$ as follows: 
for each $f_{t}=\{u_{t}, v_{t}\} \ (t = 1, \ldots, m)$, 
\begin{conference}%
\begin{align*}
w(e_{i,l+2t-2})&=\begin{cases}
t-1&(i\in [n]),\\
0  &(i\ge n+1),
\end{cases}
~\text{and}&
w(e_{i,l+2t-1})=\begin{cases}
t&(i\in\{u_{t},v_{t}\}),\\
t-1&(i\in [n]\setminus\{u_{t},v_{t}\}),\\
0&(i \ge n+1).\\
\end{cases}
\end{align*}
\end{conference}%
\begin{fullpaper}%
\begin{align*}
w(e_{i,l+2t-2})&=\begin{cases}
t-1&(i\in \{1, \ldots, n\}),\\
0  &(i\in \{n+1, \ldots, n+2m\}),\\
\end{cases}
\end{align*}
and
\begin{align*}
w(e_{i,l+2t-1})&=\begin{cases}
t&(i\in\{u_{t},v_{t}\}),\\
t-1&(i\in\{1,\dots,n\}\setminus\{u_{t},v_{t}\}),\\
0&(i \in \{n+1, \ldots, n+2m\}).\\
\end{cases}
\end{align*}
\end{fullpaper}%
\end{itemize}
The weight is illustrated in Table~\ref{table:weight}.
\begin{table}[h]
\caption{The weight of each element $e_{ij}$, where each row corresponds to $i$ and each column corresponds to $j$.}\label{table:weight}
\centering
\scalebox{0.8}{%
$\begin{array}{c||cccccccccc|ccc}
i\backslash j&1&\cdots&l-1&l&\cdots&l+2t-2&l+2t-1&l+2t&\cdots&l+2m-1&l+2m&\cdots&n+2m-1 \\ \hline\hline
1&0&\cdots&0&0&&t-1&t-1&t  &&&m&\cdots&m \\
\vdots&\vdots&&\vdots&\vdots&&\vdots&\vdots&\vdots&&&\vdots&&\vdots \\
\vdots&\vdots&&\vdots&\vdots&&t-1&t-1&t  &&&\vdots&&\vdots \\
u_t&\vdots&&\vdots&\vdots&&t-1&t  &t  &&&\vdots&&\vdots \\
\vdots&\vdots&&\vdots&\vdots&&t-1&t-1&t  &&&\vdots&&\vdots \\
\vdots&\vdots&&\vdots&\vdots&&\vdots&\vdots&\vdots&&&\vdots&&\vdots \\
\vdots&\vdots&&\vdots&\vdots&&t-1&t-1&t  &&&\vdots&&\vdots \\
v_t&\vdots&&\vdots&\vdots&&t-1&t  &t  &&&\vdots&&\vdots \\
\vdots&\vdots&&\vdots&\vdots&&t-1&t-1&t  &&&\vdots&&\vdots \\
\vdots&\vdots&&\vdots&\vdots&&\vdots&\vdots&\vdots&&&\vdots&&\vdots \\
n&0&\cdots&0&0&&t-1&t-1&t  &&&m&\cdots&m \\ \hline
n+1&0&\cdots&0&0&    \cdots&0&0&0&    \cdots&0&2m^2&\cdots&2m^2 \\
\vdots&\vdots&&\vdots&\vdots&&\vdots&\vdots&\vdots&&\vdots&\vdots&&\vdots \\
n+2m&0&\cdots&0&0&\cdots&0&0&0&\cdots&0&2m^2&\cdots&2m^2\\ \hline
\end{array}
$}
\end{table}
We remark that $|E| = (n+2m)(n+2m-1)$. 
By the definition of the matroid, for every $i \in [n+2m]$, all elements in $E_i$ belong to different independent sets from each other. 
Thus, for any feasible partition of $E$, each independent set has $n+2m-1$ elements which consist of one element from each $E_i$ except one set. 

\begin{conference}
The following claim now proves the theorem.
The proof can be found in Appendix. \qedhere
\end{conference}
\begin{fullpaper}
It remains to show that the resulting instance is equivalent to the densest $l$-subgraph problem instance $(G=(V,F), l)$. 
\end{fullpaper}
\begin{claim}\label{claim:sNP}
Let $\alpha \in \{0, \ldots, m\}$. 
The graph $G$ has a vertex set $T^*$ with $|T^*|=l$ and $|F[T^*]|\ge\alpha$ if and only if there exists a feasible partition $(I_1, \ldots, I_k)$ of $E$ with $(\sum,\max)$-value at most $2m^2(n-l)+m^2+m-\alpha$.
\end{claim}
\begin{fullpaper}
First, we assume that there exists $T^*\subseteq V$ such that $|T^*|=l$ and $|F[T^*]|\ge\alpha$. 
Without loss of generality, we assume that $T^*=\{1,\dots,l\}$ and $V\setminus T^*=\{l+1,\dots,n\}$.
We show that there exists a partition such that its $(\sum,\max)$-value is at most \(2m^2(n-l)+m^2+m-\alpha\). 
We denote $E^j[p,q]=\{e_{p,j},\dots,e_{q,j}\}$.
Let $J_1=\{1,\dots,l\}$, $J_2=\{l+1,\dots,l+2m\}$, and $J_3=\{l+2m+1,\dots,n+2m\}$.
We construct a partition $(I^*_1, \ldots, I^*_{n+2m})$ of $E$ as follows:
\begin{align*}
I_j^*&=\begin{cases}
E^{j-1}[1,j-1]\cup E^j[j+1,n+2m] &(j\in J_1),\\
E^{j-1}[1,l]\cup E^j[l+1,n+2m+l-j]\cup E^{j-1}[n+2m+l-j+2,n+2m]&(j\in J_2),\\
E^{j-1}[1,j-2m-1]\cup E^j[j-2m+1,n]\cup E^{j-1}[n+1,n+2m]&(j\in J_3).
\end{cases}
\end{align*}
Then, the maximum weight of each independent set is
\begin{align*}
\max_{e\in I_j^*}w(e)
&=\begin{cases}
0&(j\in J_1),\\
t-1&(j=l+2t-1\in J_2,~t=1,\dots,m,~ \{u_t, v_t\}\in F[T^*]),\\
t  &(j=l+2t-1\in J_2,~t=1,\dots,m,~\{ u_t, v_t\}\not\in F[T^*]),\\
t  &(j=l+2t\in J_2,~t=1,\dots,m),\\
2m^2 & (j\in J_3).\\
\end{cases}
\end{align*}
Thus, the $(\sum,\max)$-value is at most
\begin{align*}
0\cdot l+\sum_{t=1}^m (2t) - |F[T^*]|+2m^2\cdot (n-l)
\le 2m^2(n-l)+m^2+m-\alpha. 
\end{align*}

Conversely, we assume that there exists a feasible partition $(I_1,\dots,I_{k})$ of $E$ such that $\max_{e\in I_1}w(e)\le\dots\le \max_{e\in I_{k}}w(e)$, and
\begin{align*}
\sum_{j\in[k]}\max_{e\in I_j}w(e) \le 2m^2(n-l)+m^2+m-\alpha. 
\end{align*}
All elements in $E_{k}$ must be contained in different $I_j$'s from each other by definition of $(E, \mathcal{I})$. 
Hence at least $n-l$ sets contain elements $e$ with $w(e) = 2m^2$. 
If $\max_{e\in I_j}w(e) \geq 2m^2$ holds for some $j \leq l+2m$, then the objective value is at least $2m^2(n-l+1)>2m^2(n-l)+m^2+m-\alpha$. 
Thus, each of $I_{l+2m+1}, \ldots, I_{k}$ contains $2m$ elements with weight $2m^2$, and none of $I_1, \ldots, I_{l+2m}$ contains such elements. 
Let 
\begin{align*}
U=\{i \mid |E_i \cap I_j|=0 ~(\exists j\in\{l+2m+1,\dots,k\})\}.
\end{align*}
Note that $|U|=n-l$ and $U \subseteq \{1, \ldots, n\}$.
Here, we have
\begin{align*}
2m^2(n-l)+m^2+m-\alpha
&\ge \sum_{j\in [k]}\max_{e\in I_j}w(e)\\
&= 2m^2(n-l)+\sum_{j\in [l+2m]}\max_{e\in I_j}w(e). 
\end{align*}
In order to obtain a lower bound of $\sum_{j\in [l+2m]}\max_{e\in I_j}w(e)$,
we define $E' = \{e_{ij} \mid i \in U, \ j=1, \ldots, l+2m \}$. 
Let $(E', \mathcal{I}')$ be a partition matroid where $\mathcal{I}' = \{I' \mid |I' \cap E_i| \leq 1 \ (\forall i \in U) \}$. 
We observe that $\sum_{j\in[l+2m]}\max_{e\in I_j}w(e) \ge \sum_{j\in[l+2m]}\max_{e\in I_j \cap E'}w(e)$, and $(I_1 \cap E', \ldots, I_{l+2m} \cap E')$ is a feasible partition to the $(\sum,\max)$ problem instance $(E', (\mathcal{I}',w), l+2m)$.  
By Lemma \ref{lem:partition}, an optimal solution to $(E', (\mathcal{I}',w), l+2m)$ can be obtained by a greedy algorithm. 
Let $(I'_1, \ldots, I'_{l+2m})$ be an output solution of the greedy algorithm. 
Then we have 
\begin{align*}
\sum_{j\in[l+2m]}\max_{e\in I_j}w(e)
&\geq \sum_{j\in[l+2m]}\max_{e\in I'_j}w(e)
= m+\sum_{l=1}^m 2(l-1) + | \{ \{u,v\} \mid | \{u,v\} \cap U| \geq 1 \}| \\
&\geq m^2 + m - | F[V\setminus U]|. 
\end{align*}
This implies $|F[V\setminus U]| \geq \alpha$. 
Therefore, $T=V \setminus U$ is a vertex set with $|T|=l$ and $|F[T]| \geq \alpha$. 

This proves the theorem.
\end{fullpaper}%
\end{proof}

Note that the matroid $(E, \mathcal{I})$ in the above proof is graphic
because it can be seen as a matroid corresponding to a cycle with $n+2m$ vertices and
each adjacent vertices is connected by $n+2m-1$ multiple edges.
Thus, the maximum total capacity spanning tree partition problem is NP-hard.

\subsection{PTAS for the identical case}\label{sec:msM-PTAS}
In this subsection,
we provide a PTAS for the minimum $(\sum,\max)$-value matroid partitioning problem with identical matroids and weights.
This is the best possible result (unless P=NP) because the problem is strongly NP-hard as we proved in the previous subsection.

We start with the following observation, which will be also useful in Section \ref{sec:msM-alg-general}. 
\begin{proposition}
Let $(E,(\mathcal{I}_i,w_i)_{i\in[k]})$ be any instance of the minimum $(\sum,\max)$-value matroid partitioning problem, and let $(I_1^*,\dots,I_k^*)$ be an optimal solution. 
When we know $\max_{e \in I^*_i} w_i(e)$ for all $i\in[k]$, 
we can easily compute a feasible partition $(I_1,\dots,I_k)$
such that \(\sum_{i\in[k]}\max_{e\in I_i}w_i(e)\le \sum_{i\in [k]}\max_{e\in I^*_i}w_i(e)\).
\end{proposition}
\begin{proof}
The feasible partitions for matroids $(E,\mathcal{I}_i|\{e\mid w_i(e)\le \max_{e^*\in I^*_i}w_i(e^*)\})_{i\in[k]}$ satisfy the condition.
Thus, we can find one of them in polynomial time by Theorem~\ref{thm:mss}.
\end{proof}

Let $(E,(\mathcal{I},w),k)$ be a problem instance, and let $\varepsilon< 1/2$ be a positive number. 
We write $w^{\max}=\max_{e\in E}w(e)$. 
Let $(I_1^*,\dots,I_k^*)$ be an optimal solution.

The idea of the algorithm is to guess the maximum weights.
Since the number of possibilities of the maximum weights is at most $n^k$,
we can solve the problem by solving the feasibility of matroid partitioning problems $n^k$ times.
Thus, we can solve the problem efficiently when $k$ is small, but not in polynomial time.
In order to reduce the possibilities, we guess $\max_{e \in I_i^*} w(e)$ only for some $i$'s.
Without loss of generality, we assume that $\max_{e\in I_1^*}w(e)\ge\dots\ge \max_{e\in I_k^*}w(e)$.
We define a set $J=\{i_1,\dots,i_s\}$ of indices by
\[
i_j = \begin{cases}
j & (j =1, \ldots, \lfloor 1/\varepsilon^2 \rfloor), \\
\lfloor (1+\varepsilon)^t/\varepsilon^2 \rfloor & (j=\lfloor 1/\varepsilon^2 \rfloor+t, \ t=1, \ldots, \lfloor\log_{1+\varepsilon}(k \varepsilon^2)\rfloor). 
\end{cases}
\]
By definition, it holds that $1=i_1<i_2<\dots<i_s\le k$, and $s = \lfloor1/\varepsilon^2\rfloor + \lfloor\log_{1+\varepsilon}(k\varepsilon^2)\rfloor$. 
Note that for any $j =\lfloor 1/\varepsilon^2 \rfloor+t$ and $t \geq 1$, 
we have 
\begin{align*}
i_{j}-i_{j-1} 
\ge ((1+\varepsilon)^t/\varepsilon^2-1)-((1+\varepsilon)^{t-1}/\varepsilon^2)
= (1+\varepsilon)^{t-1}/\varepsilon - 1 \geq 1/\varepsilon - 1 > 1
\end{align*}
as $\varepsilon<1/2$.
For notational convenience, we denote $i_0=0$ and $i_{s+1}=k+1$.

To reduce the number of possibilities more, we round the weights $w(e)$. 
For all $e \in E$, define
\begin{align*}
w'(e)=\begin{cases}
\frac{(1+\varepsilon)^t  w^{\max}}{k}\varepsilon& \left(\frac{(1+\varepsilon)^{t}w^{\max}}{k}\varepsilon\le w(e)< \frac{(1+\varepsilon)^{t+1} w^{\max}}{k}\varepsilon, \  t\in\{0,1,\dots,\lfloor\log_{1+\varepsilon}(\frac{k}{\varepsilon})\rfloor\} \right),\\
0                      & \left( w(e)< \frac{w^{\max}}{k}\varepsilon\right).
\end{cases}
\end{align*}

Our algorithm guesses $\max_{e\in I_{i_j}^*}w'(e)$ for each $i_j\in J$.
We write $u_j^*$ for the value.
Then, it finds a feasible partition $(I_1,\dots,I_k)$ that satisfies
$\max_{e\in I_1}w(e)\ge\dots\ge \max_{e\in I_k}w(e)$ and 
$\max_{e\in I_{i_j}} w'(e)\le u_j^*$ for all $i_j\in J$.
The algorithm is summarized in Algorithm \ref{alg:msM}. 
\begin{algorithm}
  \caption{PTAS for the $(\sum,\max)$ problem with identical matroids and weights}\label{alg:msM}
  \ForEach{$u_{1},\dots,u_{s}\in\{0\}\cup\left\{\frac{(1+\varepsilon)^t w^{\max}}{k}\varepsilon\mid t=0,\dots,\lfloor\log_{1+\varepsilon}(k/\varepsilon)\rfloor\right\}$ \textrm{such that} $u_{1}\ge\dots\ge u_{s}$}{
    find a partition $(I_1,\dots,I_k)$ such that $I_i\in (\mathcal{I}|\{e\mid w'(e)\le u_j\})$ for each
    $i_j\le i<i_{j+1},~j=1,\dots,s$ if such a partition exists\;
  }
  \Return the best solution $(I_1, \ldots, I_k)$ among the obtained partitions\;
\end{algorithm}

\begin{theorem}
Algorithm \ref{alg:msM} is a PTAS algorithm for the minimum $(\sum,\max)$-value matroid partitioning problem with identical matroids and weights. 
\end{theorem}
\begin{proof}
Let $(I_1^*,\dots,I_k^*)$ be an optimal solution to the problem
and $(I_1,\dots,I_k)$ be the output of Algorithm \ref{alg:msM}.
Without loss of generality, we assume that $\max_{e\in I_1^*}w(e)\ge\dots\ge \max_{e\in I_k^*}w(e)$. 
Let $u_j^*=\max_{e\in I_{i_j}^*}w'(e)$ for each $i_j\in J$. 

We first analyze the running time of Algorithm \ref{alg:msM}.
\begin{claim}\label{claim:PTAS-time}
Algorithm \ref{alg:msM} runs in polynomial time with respect to $k$ for fixed $\varepsilon$. 
\end{claim}
\begin{proof}[Proof of Claim \ref{claim:PTAS-time}.]
Let $r=\lfloor\log_{1+\varepsilon}(k/\varepsilon)\rfloor+2$.
%
%
We observe that any choice of a possible combination of values $u_1,\dots,u_s$ corresponds a multisubset of size $s$ from the set of $r$ values. 
Thus the number of possible combinations is $\binom{r+s-1}{s}$.  
Furthermore, we have
\begin{align*}
  \binom{r+s-1}{s}
  &\le \sum_{l=0}^{r+s-1} \binom{r+s-1}{l}
  = 2^{r+s-1}
  \le 2^{(\log_{1+\varepsilon}(k/\varepsilon) + 2) + (1/\varepsilon^2 + \log_{1+\varepsilon}(k\varepsilon^2))}\\
  &\le 2^{2\log_{1+\varepsilon}k + 2 + 1/\varepsilon^2}
  = 2^{2+1/\varepsilon^2}\cdot k^{\log_{1+\varepsilon}4}.
\end{align*}
This is a polynomial with respect to $k$ for fixed $\varepsilon$.
Thus, the algorithm runs in polynomial time.
\end{proof}
Note that, without the restriction $u_{1}\ge\dots\ge u_{s}$,
the number of possible combinations of values $u_1,\dots,u_s$ is $r^s=k^{\Theta(\log\log k)}$, which is not polynomial with respect to $k$.

In the remainder, we show the approximation ratio of the algorithm.
\begin{claim}\label{claim:PTAS-approx}
Let $\OPT$ denote the optimal value and let $\ALG$ denote the $(\sum, \max)$-value of $(I_1, \ldots, I_k)$. 
Then it holds that $\ALG \leq (1+15.5\varepsilon)\OPT$. 
\end{claim}
\begin{proof}[Proof of Claim \ref{claim:PTAS-approx}.]
First, $\OPT$ is at least
\begin{conference}
\(\OPT=\sum_{i\in [k]} \max_{e\in I_i^*}w(e) \ge \sum_{i\in [k]} \max_{e\in I_i^*}w'(e) \ge \sum_{j=1}^s (i_j-i_{j-1})u_j^*\).
\end{conference}
\begin{fullpaper}
\begin{align*}
  \OPT=\sum_{i\in [k]} \max_{e\in I_i^*}w(e)
  \ge \sum_{i\in [k]} \max_{e\in I_i^*}w'(e)
  \ge \sum_{j=1}^s (i_j-i_{j-1})u_j^*.
\end{align*}
\end{fullpaper}
Let $(I'_1, \ldots, I'_k)$ be a feasible partition of $E$ obtained at line 2 in Algorithm \ref{alg:msM} using $u_1^*, \ldots, u_s^*$. 
Then $\ALG$ is at most
\begin{conference}
\begin{align}
  \ALG
  &=  \sum_{i\in [k]} \max_{e \in I_i} w(e)  \leq \sum_{i\in[k]} \max_{e \in I'_i} w(e) \notag\\
  &\le \sum_{j=1}^s (i_{j+1}-i_{j}) \max_{e \in I'_{i_j}} w(e) 
  \le \sum_{j=1}^s (i_{j+1}-i_{j}) \left((1+\varepsilon)u_j^* + \frac{w^{\max}}{k}\varepsilon \right) \notag\\
  & \le \sum_{j=1}^s (i_{j+1}-i_{j})(1+\varepsilon)u_j^*+k\cdot\frac{w^{\max}}{k}\varepsilon
  \le (1+\varepsilon)\sum_{j=1}^s (i_{j+1}-i_{j})u_j^*+\varepsilon\cdot \OPT. \label{eq:msM PTAS 0}
\end{align}
\end{conference}
\begin{fullpaper}
\begin{align}
\ALG &=  \sum_{i\in [k]} \max_{e \in I_i} w(e)  \leq \sum_{i\in[k]} \max_{e \in I'_i} w(e) \notag\\
     & \le \sum_{j=1}^s (i_{j+1}-i_{j}) \max_{e \in I'_{i_j}} w(e) \notag\\
     & \le \sum_{j=1}^s (i_{j+1}-i_{j}) \left((1+\varepsilon)u_j^* + \frac{w^{\max}}{k}\varepsilon \right) \notag\\
     & \le \sum_{j=1}^s (i_{j+1}-i_{j})(1+\varepsilon)u_j^*+k\cdot\frac{w^{\max}}{k}\varepsilon
       \le (1+\varepsilon)\sum_{j=1}^s (i_{j+1}-i_{j})u_j^*+\varepsilon\cdot \OPT. \label{eq:msM PTAS 0}
\end{align}
\end{fullpaper}
Here, the third inequality holds by the definition of $w'$ and $\max_{e \in I'_{i_j}} w'(e)\le u_j^*$.

We derive an upper bound on $\sum_{j=1}^s (i_{j+1}-i_{j})u^*_j$. 
To simplify notation, let $q = \lfloor 1/\varepsilon^2 \rfloor$. 
First, since $i_{j+1}-i_{j} = i_{j}-i_{j-1} = 1$ holds for any $j =1, \ldots, q -1$, we have 
\begin{align}\label{eq:msM PTAS 1}
\begin{conference}\textstyle\end{conference}\sum_{j=1}^{q -1} (i_{j+1}-i_{j})u^*_j = \sum_{j=1}^{q-1} (i_{j}-i_{j-1})u^*_j.
\end{align}

Second, we evaluate $(i_{q+1}-i_q)u^*_q$. 
Note that $i_q = q = \lfloor 1/\varepsilon^2 \rfloor$ and $i_{q+1} = \lfloor (1+\varepsilon)/\varepsilon^2 \rfloor$. 
Thus $i_{q+1}-i_q \leq (1+\varepsilon)/\varepsilon^2 - (1/\varepsilon^2 - 1) = (1+\varepsilon)/\varepsilon$. 
Moreover, 
\begin{conference}
$u^*_q = \max_{e \in I^*_{q}} w'(e) \leq \max_{e\in I^*_q} w(e) \leq \OPT/q$,
\end{conference}
\begin{fullpaper}
\begin{align*}
u^*_q = \max_{e \in I^*_{q}} w'(e) \leq \max_{e\in I^*_q} w(e) \leq \OPT/q, 
\end{align*}
\end{fullpaper}
because $\OPT = \sum_{i\in[k]} \max_{e\in I^*_i} w(e) \geq \sum_{i\in[q]} \max_{e\in I^*_i} w(e) \geq q \cdot \max_{e\in I^*_q} w(e)$. 
We remark that 
$1/q = 1/\lfloor 1/\varepsilon^2 \rfloor \le 1/(1/\varepsilon^2-1)=\varepsilon^2/(1-\varepsilon^2) < \frac{4}{3}\varepsilon^2 < 2\varepsilon^2$ as $\varepsilon < 1/2$.
Therefore, it follows that 
\begin{align}\label{eq:msM PTAS 2}
(i_{q+1}-i_q)u^*_q \leq 2\varepsilon(1+\varepsilon) \OPT. 
\end{align}

Lastly, let $j \in \{q+1, \ldots, s\}$, and let $t~(\ge 1)$ be the integer such that $i_j = \lfloor (1+\varepsilon)^t/\varepsilon^2 \rfloor$ (i.e., $t=j-q$). 
We observe that $i_{j}-i_{j-1} \geq (1+\varepsilon)^{t-1}/\varepsilon -1$. 
In addition, we have
\begin{conference}
\begin{align*}
i_{j+1}-i_{j}
&\le \left(\frac{(1+\varepsilon)^{t+1}}{\varepsilon^2}\right) - \left(\frac{(1+\varepsilon)^{t}}{\varepsilon^2}-1\right)
=\frac{(1+\varepsilon)^{t}}{\varepsilon}+1 \\
&\le \frac{(1+\varepsilon)/\varepsilon+1}{(1+\varepsilon)^0/\varepsilon-1}\left(\frac{(1+\varepsilon)^{t-1}}{\varepsilon}-1\right)
\le \frac{1+2\varepsilon}{1-\varepsilon}(i_{j}-i_{j-1})
< (1+6\varepsilon)(i_{j}-i_{j-1}),
\end{align*}%
\end{conference}%
\begin{fullpaper}
\begin{align*}
i_{j+1}-i_{j}
&\le \left(\frac{(1+\varepsilon)^{t+1}}{\varepsilon^2}\right) - \left(\frac{(1+\varepsilon)^{t}}{\varepsilon^2}-1\right)
=\frac{(1+\varepsilon)^{t}}{\varepsilon}+1 \\
&\le \frac{(1+\varepsilon)/\varepsilon+1}{(1+\varepsilon)^0/\varepsilon-1}\left(\frac{(1+\varepsilon)^{t-1}}{\varepsilon}-1\right)\\
&\le \frac{1+2\varepsilon}{1-\varepsilon}(i_{j}-i_{j-1})
< (1+6\varepsilon)(i_{j}-i_{j-1}),
\end{align*}
\end{fullpaper}%
where the second inequality holds since $\frac{(1+\varepsilon)^x/\varepsilon+1}{(1+\varepsilon)^{x-1}/\varepsilon-1}$
is monotone decreasing for $x\ge 1$ and
the last inequality holds since $\varepsilon < 1/2$. 
Therefore, it follows that 
\begin{align}\label{eq:msM PTAS 3}
\begin{conference}\textstyle\end{conference}\sum_{j=q+1}^{s} (i_{j+1}-i_{j})u^*_j = \sum_{j=q+1}^{s} (1+6\varepsilon)(i_{j}-i_{j-1})u^*_j. 
\end{align}

By combining \eqref{eq:msM PTAS 0}, \eqref{eq:msM PTAS 1}, \eqref{eq:msM PTAS 2}, \eqref{eq:msM PTAS 3}, together with $\varepsilon < 1/2$, we have
\begin{conference}
\begin{align*}
\ALG &\leq (1+\varepsilon)\left( (1+6\varepsilon) \sum_{j=1}^s (i_{j}-i_{j-1})u^*_j + 2\varepsilon(1+\varepsilon) \OPT \right) +\varepsilon\cdot \OPT \\
&\leq (1+\varepsilon)\left( (1+6\varepsilon) + 2\varepsilon(1+\varepsilon) \right) \cdot \OPT +\varepsilon\cdot \OPT 
=(1+10\varepsilon + 10\varepsilon^2 + 2\varepsilon^3)\OPT \\
&< (1+10\varepsilon + 5\varepsilon + 0.5\varepsilon)\OPT
=(1+15.5\varepsilon)\OPT. & \qedhere
\end{align*}
\end{conference}%
\begin{fullpaper}
\begin{align*}
\ALG &\leq (1+\varepsilon)\left( (1+6\varepsilon) \sum_{j=1}^s (i_{j}-i_{j-1})u^*_j + 2\varepsilon(1+\varepsilon) \OPT \right) +\varepsilon\cdot \OPT \\
&\leq (1+\varepsilon)\left( (1+6\varepsilon) + 2\varepsilon(1+\varepsilon) \right) \cdot \OPT +\varepsilon\cdot \OPT \\
&=(1+10\varepsilon + 10\varepsilon^2 + 2\varepsilon^3)\OPT \\
&< (1+10\varepsilon + 5\varepsilon + 0.5\varepsilon)\OPT
=(1+15.5\varepsilon)\OPT. & \qedhere
\end{align*}
\end{fullpaper}%
\end{proof}\let\qed\relax
\begin{fullpaper}
Claims \ref{claim:PTAS-time} and \ref{claim:PTAS-approx} imply that Algorithm \ref{alg:msM} is a PTAS.
\end{fullpaper}
\end{proof}

\subsection{Hardness of the general case}
We show a stronger result than the NP-hardness of the minimum $(\sum,\max)$-value matroid partitioning problem by reducing the \emph{set cover} problem. 
Given a set $V=[n]$ and a collection $\mathcal{S} = \{ S_i \subseteq V \mid i\in [k] \}$, the set cover problem is to find a subset $\mathcal{S}'~(\subseteq\mathcal{S})$ of minimum cardinality
such that $\mathcal{S}'$ covers $V$, i.e., $\bigcup_{S\in\mathcal{S}'}S=V$.
It is known that the set cover problem cannot be approximated in polynomial time to within a factor of $o(\log k)$ unless P=NP~\cite{DinurSteurer2014,Moshkovitz2015}. 
\begin{theorem}
Even if either matroids or weights, but not both, are identical,
the minimum $(\sum,\max)$-value matroid partitioning problem cannot be approximated in polynomial time
within a factor of $o(\log k)$, unless P=NP.
\end{theorem}
\begin{proof}
Let $(V, \mathcal{S})$ be an instance of the set cover problem. 
We first construct an instance of the minimum $(\sum,\max)$-value matroid partitioning problem with identical matroids and different weights. 
Define a set $D$ of $(k-1)n$ dummy elements. 
Let $E=V \cup D$ be the ground set and let $\mathcal{I}=\{I\subseteq E\mid |I|\le n\}$.
For each element $e \in E$ and $i\in [k]$, the weight $w_i(e)$ is defined by
$0$ if $e\in D$, $1$ if $e\in S_i$, and a sufficiently large number if $e\in V\setminus S_i$ (for example, $k^2$).
Let us consider an instance $(E,(\mathcal{I},w_i)_{i\in [k]})$ of the minimum $(\sum,\max)$-value matroid partitioning problem,
which is the identical matroids case.

To show the theorem for the identical matroid case, it is sufficient to prove that there exists a set cover $\mathcal{S}^*~(\subseteq \mathcal{S})$ of size at most $t$ if and only if there exists a feasible partition $(I_1, \ldots, I_k)$ of $E$ such that $\sum_{i\in [k]} \max_{e \in I_i} w_i(e)\leq t$. 

First, let $\mathcal{S}^*$ be a set cover of size $t$. 
For notational convenience, we denote $\mathcal{S}^* = \{S_1, \ldots, S_t\}$. 
By the definition of $S_i$'s, there exists some partition $(V_1, \ldots, V_t)$ of $V$ such that $V_i \subseteq S_i \ (i\in [t])$. 
Some $V_i$'s may be empty. 
Then, let $(D_1, \ldots, D_k)$ be an arbitrary partition of $D$ such that $|D_i| = n - |V_i|$ for $i=1, \ldots, t$ and $|D_i| = n$ for $i=t+1,\dots,k$. 
We construct a partition $(I_1, \ldots, I_k)$ of $E$ defined by $I_i = V_i \cup D_i$ if $i \leq t$ and $I_i = D_i$ if $i > t$. 
Since $|I_i| = n$ for all $i$, this partition is feasible. 
At most $t$ sets in $I_i$'s contain an element of $V$, and hence it holds that $\sum_{i\in[k]} \max_{e \in I_i} w_i(e)\leq t$. 

Conversely, let $(I_1, \ldots, I_k)$ be a partition of $E$ such that $\sum_{i\in[k]} \max_{e \in I_i} w_i(e)= t~(\le k)$. 
Here, the value of $\max_{e \in I_i} w_i(e)$ is $1$ if $I_i$ contains an element of $V$ and otherwise the value is $0$.
Since every element of $V$ is contained in some $I_i$, $t$ sets among $I_1, \ldots, I_k$ contain an element of $V$. 
We denote such sets by $I_1, \ldots, I_t$ by rearranging the indices. 
Since $\bigcup_{i\in [t]} I_i \setminus D = V$ and $I_i \setminus D \subseteq S_i \ (i\in[k])$,
we have $\bigcup_{i\in [t]} S_i \supseteq V$. 
Thus, $( S_1, \ldots, S_t )$ is a set cover of size $t$. 
This proves the theorem for identical matroid case. 

For the identical weights case, 
we construct an instance $(E,(\mathcal{I}_i,w)_{i\in [k]})$ of the minimum $(\sum,\max)$-value matroid partition problem by setting 
\begin{conference}%
$\mathcal{I}_i = \{ I\subseteq S_i \cup D \mid |I| \leq n\}$,
$w(e)=0$ for $e\in D$, and $w(e)=1$ for $e\in V$.
\end{conference}
\begin{fullpaper}
\begin{align*}
  \mathcal{I}_i = \{ I\subseteq S_i \cup D \mid |I| \leq n\}
  \quad\text{and}\quad
  w(e)=\begin{cases}
  0&(e\in D),\\
  1&(e\in V).
  \end{cases}  
\end{align*}
\end{fullpaper}
Then the statement holds by a similar proof. 
\end{proof}

\subsection{Algorithm for the general case}\label{sec:msM-alg-general}
In this subsection, we provide an $\varepsilon k$-approximation algorithm for any $\varepsilon>0$.
Let $(E,(\mathcal{I}_i,w_i)_{i\in[k]})$ be an instance of the minimum $(\sum,\max)$-value matroid partitioning problem, and let $(I_1^*,\dots,I_k^*)$ be any optimal partition. 

Similarly to the PTAS described in Section \ref{sec:msM-PTAS}, our algorithm guesses $\max_{e\in I^*_{i}}w_i(e)$ for each $i\in [k]$.
In order to reduce the number of possibilities,
we only guess top-$\lceil 1/\varepsilon\rceil$ weights of $\max_{e \in I_i^*} w_i(e)$.
For simplicity, let $r=\lceil 1/\varepsilon\rceil$.
Let $J^*=\{i_1,\dots,i_r\}$ be the indices of top-$r$ weights, i.e.,
$\max_{e\in I_{i}^*}w_i(e)\ge \max_{e\in I_{j}^*}w_i(e)$ for any $i\in J^*$ and $j\in [k]\setminus J^*$.
Let $u_i^*=\max_{e\in I_{i}^*}w_i(e)$ for each $i\in J^*$.
Then it finds a feasible partition $(I_1,\dots,I_k)$ that satisfies
$\max_{e\in I_i} w_i(e) \le u_i^*$ for $i\in J^*$ and
$\max_{e\in I_i} w_i(e) \le \min_{j\in J^*}u_j^*$ for $i\in [k]\setminus J^*$.

\begin{conference}
\begin{theorem}\label{thm:summax-approx}
  For any positive fixed number $\varepsilon > 0$, there exists a polynomial-time $\varepsilon k$-approximation algorithm for the minimum $(\sum,\max)$-value matroid partitioning problem.
\end{theorem}
The proof can be found in Appendix.
\end{conference}

\begin{fullpaper}
The algorithm is summarized in Algorithm \ref{alg:msM_general}. 
\begin{algorithm}
  \caption{$\varepsilon k$-approximation for the $(\sum,\max)$ problem}\label{alg:msM_general}
  \ForEach{$J\subseteq [k]$ such that $|J|=r$}{
    \ForEach{$u_i\in\{w_i(e)\mid e\in E\} $ $(i\in J)$}{
      find a partition $(I_1,\dots,I_k)$ such that
      $I_i\in (\mathcal{I}|\{e\mid w_i(e)\le u_i\})$ for $i\in J$ and
      $I_i\in (\mathcal{I}|\{e\mid w_i(e)\le \min_{j\in J}u_j\})$ for $i\in [k]\setminus J$       
      if such a partition exists\;
    }
  }
  \Return the best solution $(I_1, \ldots, I_k)$ among the obtained partitions\;
\end{algorithm}

\begin{theorem}
For any positive fixed number $\varepsilon > 0$, Algorithm \ref{alg:msM_general} is a polynomial-time $\varepsilon k$-approximation algorithm for the minimum $(\sum,\max)$-value matroid partitioning problem.
\end{theorem}
\begin{proof}
Let $(I_1^*,\dots,I_k^*)$ be an optimal solution to the problem
and $(I_1,\dots,I_k)$ be the output of Algorithm \ref{alg:msM_general}.

We first analyze the running time of Algorithm \ref{alg:msM_general}. 
The number of possibility of $J$ is $\binom{k}{r}~(\le k^r)$. 
For each $J$, the number of possible combination of values $u_i$ for $i\in J$ is $|E|^r$.
Hence, the algorithm checks the feasibility of the matroid partitioning problem at most $(k|E|)^{r}=(k|E|)^{2\lceil 1/\varepsilon\rceil}$ times
and this is a polynomial with respect to $k$ and $|E|$ for fixed $\varepsilon$.
Thus, the algorithm runs in polynomial time.

Next, we show the approximation ratio of the algorithm.
The optimum value $\OPT$ is at least
\begin{align*}
  \OPT=\sum_{i\in[k]}\max_{e\in I_i^*}w_i(e)
  \ge \sum_{i\in J^*}\max_{e\in I_i^*}w_i(e)
  =\sum_{i\in J^*}u_i^*
  \ge r\cdot \left(\min_{j\in J^*}u_j^* \right).
\end{align*}
Let $(I_1',\dots,I_k')$ is a feasible partition of $E$ obtained at line 3 in Algorithm \ref{alg:msM_general}
using $J=J^*$ and $u_i=u_i^*$ for each $i\in J$.
Then, the objective value $\ALG$ of $(I_1,\dots,I_k)$ is at most
\begin{align*}
  \ALG
  &=\sum_{i\in[k]}\max_{e\in I_i}w_i(e)
  \le\sum_{i\in[k]}\max_{e\in I_i'}w_i(e)
  \le\sum_{i\in J^*}u_{i}^*+(k-r)\left(\min_{j\in J^*}u_j^*\right)\\
  &\le \OPT+\frac{k-r}{r}\OPT=\frac{k}{r}\OPT\le \varepsilon k\cdot\OPT.
\end{align*}
Thus, it is an $\varepsilon k$-approximation algorithm.
\end{proof}
\end{fullpaper}


%% file: tractable.tex
\section{Polynomial-time solvable optimal matroid partitioning problems}\label{sec:tractable}
In this section, we provide algorithms for the cases (1) $(\outerop, \innerop) = (\min,\min)$, $(\max,\max)$, $(\min,\max)$ or $(\min,\sum)$; (2) $(\outerop, \innerop) = (\max,\min)$ or $(\sum,\min)$ with identical matroids. 

We first deal with the first case. 
For the $(\min,\min)$, $(\max,\max)$, $(\min,\max)$ and $(\min,\sum)$ problems, we show polynomial-time reductions to the matroid partitioning problem. 
Then we can see that these are polynomially solvable by Theorem \ref{thm:mss}.
\begin{theorem}\label{thm:mmm}
The minimum $(\min,\min)$-value matroid partitioning problem is solvable in polynomial time.
\end{theorem}
\begin{fullpaper}
\begin{proof}
Let $(E,(\mathcal{I}_i,w_i)_{i\in [k]})$ be any problem instance. 
We denote by $(I_1^*,\dots,I_k^*)$ an optimal feasible partition of the problem. 
Define \((i^*,e^*)\in \argmin_{(i,e)\in[k]\times I_i^*}w_i(e)\).
Note that the $(\min,\min)$-value of $(I_1^*,\dots,I_k^*)$ is $w_{i^*}(e^*)$.
Let $(i',e')\in[k]\times E$ be a pair that attains the minimum of $w_i(e)$ among $(i,e)\in[k]\times E$ such that
\begin{align}
  E\setminus\{e'\}\text{ has a feasible partition for }
  (E \setminus \{e'\}, \mathcal{I}_i)_{i \neq i'} \ \text{and} \ (E \setminus \{e'\}, (\mathcal{I}_{i'}/\{e'\})).
  \label{eq:mmm_condition}
\end{align}
The number of possibility of $(i,e)$ is $k\cdot |E|$. 
For each $(i,e)$, the condition \eqref{eq:mmm_condition} can be checked in polynomial time by Theorem \ref{thm:mss}.
Thus we can find $(i',e')$ in polynomial time. 
Then, the optimal value \(w_{i^*}(e^*)\) is at most \(w_{i'}(e')\), because there exists a feasible partition $(I_1',\dots,I_k')$
such that $e'\in I_{i'}'$.
On the other hand, the optimal value \(w_{i^*}(e^*)\) is at least \(w_{i'}(e')\) because
\(E\setminus\{e^*\}\in \mathcal{I}_1\vee\dots\vee\mathcal{I}_{i^*-1}\vee(\mathcal{I}_{i^*}/\{e^*\})\vee\mathcal{I}_{i^*+1}\vee\dots\vee\mathcal{I}_{k}\).
Thus, $w_{i'}(e')$ is the optimal value of the problem, and hence we can compute the optimal value (and partition) in polynomial time.
\end{proof}

We also show the polynomial-time solvability of the minimum $(\max,\max)$, and $(\min,\max)$-value matroid partitioning problems in a similar way. 
\end{fullpaper}
\begin{conference}
We also show the polynomial-time solvability of the minimum $(\max,\max)$, and $(\min,\max)$-value matroid partitioning problems in a similar way; see Theorem \ref{thm:max-max,min-max} in Appendix.  
\end{conference}

\begin{theorem}\label{thm:max-max,min-max}
The minimum $(\max,\max)$ and $(\min,\max)$-value
matroid partitioning problems $(E,(\mathcal{I}_i,w_i)_{i\in [k]})$ are solvable in polynomial time.
\end{theorem}
\begin{fullpaper}
	\begin{proof}
		We prove the $(\max,\max)$ case.
		The key idea is that the optimal value is at most $w$ if and only if
		\begin{align}
		E\text{ has a feasible partition for }
		(\mathcal{I}_1|\{e:w_1(e)\le w\},\dots,\mathcal{I}_{k}|\{e:w_k(e)\le w\}).
		\label{eq:mMM_condition}
		\end{align}
		Since the optimal value is in \(\{w_i(e)\mid i\in [k],~e\in E\}\),
		we can obtain the optimal value by setting $w$ for all the possibilities.
		As the condition \eqref{eq:mMM_condition} can be checked in polynomial time by Theorem \ref{thm:mss}
		and the number of the possibilities is at most $k\times |E|$,
		we can compute the optimal value (and partition) in polynomial time.
		
		Next, we see the $(\min,\max)$ case.
		In this case, the optimal value is at most $w$ if and only if
		there exists $i^*$ such that
		\begin{align}
		E\text{ has a feasible partition for }
		(\mathcal{I}_1,\dots,\mathcal{I}_{i^*-1},(\mathcal{I}_{i^*}|\{e:w_{i^*}(e)\le w\}),\mathcal{I}_{i^*+1},\dots,\mathcal{I}_{k}).
		\label{eq:mmM_condition}
		\end{align}
		Thus, we can find the optimal value by setting $i^*$ for all \([k]\)
		and $w$ for all \(\{w_{i^*}(e)\mid e\in E\}\).
	\end{proof}
\end{fullpaper}

\begin{theorem}\label{thm:mms}
The minimum $(\min,\sum)$-value matroid partitioning problem $(E,(\mathcal{I}_i,w_i)_{i\in [k]})$ is solvable in polynomial time.
\end{theorem}
\begin{fullpaper}
\begin{proof}
Let $(E,(\mathcal{I}_i,w_i)_{i\in [k]})$ be any problem instance. 
We denote by $(I_1^*,\dots,I_k^*)$ an optimal partition of the problem. 
For each $j\in[k]$, let $(I_1^{j},\dots,I_k^{j})$ be an optimal solution for a minimum $(\sum,\sum)$-value of
the matroid partitioning problem instance $(E,(\mathcal{I}_i,w_i^{j})_{i\in [k]})$, 
where 
\begin{align*}
  w_i^{j}(e)=\begin{cases}
  w_i(e)&(i=j),\\
  0&(i\ne j).
  \end{cases}
\end{align*}
Then, we have
\begin{align*}
  \min_{i\in[k]}\sum_{e\in I_i^*} w_i(e)
  =\min_{j\in[k]}\sum_{i\in[k]}\sum_{e\in I_i^*} w_i^{j}(e)
  =\min_{j\in[k]}\sum_{i\in[k]}\sum_{e\in I_i^{j}} w_i^{j}(e).
\end{align*}
Thus, we can obtain an optimal solution by solving $(E,(\mathcal{I}_i,w_i^{j})_{i\in [k]})$ 
for all $j\in [k]$.
The running time is polynomial by using the polynomial-time algorithm in Theorem \ref{thm:mss}.
\end{proof}
\end{fullpaper}
\begin{conference}
The proofs of the theorems can be found in Appendix.
\end{conference}

Next we consider the $(\max,\min)$ case and the $(\sum,\min)$ case.
As we will see in the next section,
the optimal matroid partitioning problems for these cases are
(strongly) NP-hard even to approximate. 
We provide polynomial-time algorithms for instances where matroids are identical (weights may differ). 
The following lemma plays the crucial role for this purpose.
\begin{lemma}\label{lem:identical_dist}
Let $(E,\mathcal{I})$ be a matroid.
If there is a partition \((I_1,\dots,I_k)\) of $E$ such that \(I_i\in\mathcal{I}\) for all \(i\in [k]\),
then for any $k$ elements \(e_1,\dots,e_k\in E\),
there is a partition \((I_1',\dots,I_k')\) of $E$ such that \(e_i\in I_i'\in\mathcal{I}\) for all \(i\in [k]\),
\end{lemma}
\begin{fullpaper}
\begin{proof}
Suppose the contrary that there exists no such partition.
Let \((I_1,\dots,I_k)\in\mathcal{I}^k\) be a partition of $E$ such that 
\(e_i\in I_i\) for all \(i\in [j]\) with $j~(<k)$ as large as possible.
Note that $e_{j+1}\not\in I_{j+1}$.
If $e_{j+1}\in I_i$ for some $i>j$,
then we can obtain a partition with larger $j$ by just swapping $I_{j+1}$ and $I_i$, which contradicts the maximality of $j$.
Otherwise, i.e., $i\le j$, there exists \(e\in I_{j+1}\) such that \((I_i\setminus\{e_{j+1}\})\cup\{e\}\in\mathcal{I}\) and \((I_{j+1}\setminus\{e\})\cup\{e_{j+1}\}\in\mathcal{I}\).
Thus, \((I_1,\dots,(I_i\setminus\{e_{j+1}\})\cup\{e\},\dots,(I_{j+1}\setminus\{e\})\cup\{e_{j+1}\},\dots,I_k)\)
is also a feasible partition and this is a contradiction.
\end{proof}
\end{fullpaper}

We will reduce the problem of finding an optimal partition to the minimum weight perfect bipartite matching problem. 
It is well-known that this problem is solvable in polynomial time (see e.g., \cite{KorteVygen2002,Schrijver2003} for basic algorithms). 
Now we are ready to prove the theorem.
\begin{theorem}
The minimum $(\max,\min)$ and $(\sum,\min)$-value
matroid partitioning problems with identical matroids $(E,(\mathcal{I},w_i)_{i\in [k]})$ 
are solvable in polynomial time.
\end{theorem}
\begin{proof}
Let $(E,\mathcal{I})$ be any matroid. 
Recall that the existence of a feasible partition is checkable in polynomial time by Theorem \ref{thm:mss}.
Hence, in what follows, we assume that $(E,(\mathcal{I},w),k)$ has a feasible partition.

We first consider the $(\max,\min)$ problem.
By Lemma \ref{lem:identical_dist},
the minimum $(\max,\min)$-value is at most $w$ if and only if
the bipartite graph $(E,[k],\{(e,i)\mid w_i(e)\le w\})$ has a right-perfect matching.
Thus, we can get the optimal value in polynomial time by setting $w$ for all \(\{w_{i}(e)\mid i\in[k],~e\in E\}\)
and checking the existence of a right-perfect matching. 

Next, we consider the $(\sum,\min)$ problem.
By Lemma \ref{lem:identical_dist},
the minimum $(\sum,\min)$-value is the minimum weight of right-perfect matchings
in the weighted bipartite graph $(E,[k],E\times[k];w)$, where weight $w$ is defined as $w(e, i)=w_i(e)$ for each $(e, i) \in E \times [k]$.
Thus, we can find the optimal value in polynomial time. 
\end{proof}

%% file: intractable.tex
\section{Hardness of optimal matroid partitioning problems}
\label{sec:intractable}

In this section,
we show that 
the minimum $(\max,\min)$ and $(\sum,\min)$-value matroid partitioning problems are both strongly NP-hard even to approximate.
We give a reduction from \emph{SAT}, which is an NP-complete problem \cite{garey1979cai}. 
\begin{theorem}
The minimum $(\max,\min)$ and $(\sum,\min)$-value matroid partitioning problems are both strongly NP-hard.
Moreover, there exists no approximation algorithm for the problems unless P=NP.
\end{theorem}
\begin{proof}
Let $U$ be the set of variables and $\mathcal{C}$ be the set of clauses for a given SAT instance.
For each clause $C\in\mathcal{C}$, let $P_C$ and $N_C$ be the set of the variables appearing in the clause $C$
as positive and negative literals, respectively.
Let $\mathcal{C}_x=\{C_x^1,\dots,C_x^{s(x)}\}\subseteq\mathcal{C}$ be the set of clauses in which $x$ occurs,
where $s(x)$ is the cardinality of $\mathcal{C}_x$.

We consider an instance of the minimum $(\max,\min)$-value matroid partitioning problem on matroids \((E,\mathcal{I}_C)~(C\in\mathcal{C})\) and \((E,\mathcal{I}_x)~(x\in U)\).
Here the ground set $E$ is given by $E=\bigcup_{x\in U}E_x$, where 
\begin{conference}%
\(E_x=\{x_d\}\cup\bigcup\nolimits_{C\in \mathcal{C}_x}\{x^C,\bar{x}^C\}\).
\end{conference}%
\begin{fullpaper}%
\begin{align*}
  E_x=\{x_d\}\cup\bigcup_{C\in \mathcal{C}_x}\{x^C,\bar{x}^C\}.
\end{align*}
\end{fullpaper}%
Intuitively, $x^C$ and $\bar{x}^C$ represent a truth assignment for $x$,
and $x_d$ is just a dummy element to keep the minimum weight corresponding to $(E,\mathcal{I}_x)$ to zero.
For each $C\in\mathcal{C}$, the independence family $\mathcal{I}_C$ is given by
\begin{align*}
  \mathcal{I}_C&=\left\{\begin{conference}\textstyle\end{conference}
  X\subseteq \bigcup_{x\in P_C\cup N_C}\{x^C,\bar{x}^C\}\mid |X\cap \{x^C,\bar{x}^C\}|\le 1~(\forall x\in P_C\cup N_C)\right\},
\end{align*}
and for each $x\in U$, the family $\mathcal{I}_x$ is defined by
\begin{align*}
  \mathcal{I}_x&=\left\{X\subseteq E_x\mid
    |X\cap\{x^{C_x^i},\bar{x}^{C_x^{i+1}}\}|\le 1~(\forall i\in\{1,\dots,s(x)\})
  \right\},
\end{align*}
where \(C_x^{s(x)+1}\) is regarded as \(C_x^1\).
We set the weights of the elements as 
\begin{align*}
  w(e)&=\begin{cases}
  0&(e\in \bigcup_{C\in \mathcal{C}}(\bigcup_{x\in P_C}\{x^C\}\cup\bigcup_{x\in N_C}\{\bar{x}^C\})\cup \bigcup_{x\in U}\{x_d\}),\\
  1&(e\in \bigcup_{C\in \mathcal{C}}(\bigcup_{x\in N_C}\{x^C\}\cup\bigcup_{x\in P_C}\{\bar{x}^C\})).\\
  \end{cases}
\end{align*}
\begin{conference}
The theorem for the $(\max, \min)$ case now follows from the following claim. 
The proof can be found in Appendix.
\begin{claim}\label{claim:NPapprox}
There exists a feasible partition whose $(\max, \min)$-value is $0$ if and only if the given SAT instance is satisfiable.
\end{claim}
\end{conference}
\begin{fullpaper}
Then, we claim that there exists a feasible partition whose $(\max, \min)$-value is $0$ if and only if the given SAT instance is satisfiable.

Assume that the SAT instance is satisfiable.
Let $\psi: U\to\{\T,\F\}$ be a truth assignment that satisfies the instance.
Then, we define a partition of $E$ by 
\begin{align*}
  I_C&=\{x^C\mid \psi(x)=\T~(x\in P_C\cup N_C)\}\cup\{\bar{x}^C\mid \psi(x)=\F~(x\in P_C\cup N_C)\}&(C\in\mathcal{C}),\\
  I_x&=\{\bar{x}^C\mid \psi(x)=\T~(C\in \mathcal{C}_x)\}\cup\{x^C\mid \psi(x)=\F~(C\in \mathcal{C}_x)\}\cup\{x_d\}&(x\in U). 
\end{align*}
It is not difficult to see that this is feasible. 
Moreover, its $(\max, \min)$-value is $0$.

Conversely, assume that there is a feasible partition problem whose $(\max, \min)$-value is $0$. 
Let $I_C^*~(C\in\mathcal{C})$ and $I_x^*~(x\in U)$ consist an optimal partition.
Note that $I_x^*$ is either \(\{x^C\mid C\in \mathcal{C}_x\}\cup\{x_d\}\) or \(\{\bar{x}^C\mid C\in \mathcal{C}_x\}\cup\{x_d\}\).
Indeed, $x^{C_x^i}\in I_x$ implies $\bar{x}^{C_x^{i+1}}\in\mathcal{C}_{C_x^{i+1}}$ and $x^{C_x^{i+1}}\in I_x$
(also $\bar{x}^{C_x^i}\in I_x$ implies $x^{C_x^{i+1}}\in\mathcal{C}_{C_x^{i+1}}$ and $\bar{x}^{C_x^{i+1}}\in I_x$).
Moreover, as the $(\max, \min)$-value is $0$, at least one of \(\bigcup_{x\in P_C}\{x^C\}\cup\bigcup_{x\in N_C}\{\bar{x}^C\}\)
is contained in $I_C$ for each $C\in\mathcal{C}$.
Thus, the following truth assignment satisfies the SAT instance:
\begin{align*}
  \psi^*(x)=\begin{cases}
  \T&(I_x^*=(\{\bar{x}^C\mid C\in \mathcal{C}_x\}\cup\{x_d\})),\\
  \F&(I_x^*=(\{x^C\mid C\in \mathcal{C}_x\}\cup\{x_d\})).
  \end{cases}
\end{align*}

Therefore, the minimum $(\max,\min)$-value matroid partition problem is NP-hard, and there exists no approximation algorithm to the problem unless P=NP.
\end{fullpaper}

Since for any feasible partition, the $(\max, \min)$-value is $0$ if and only if the $(\sum, \min)$-value is $0$, the proof works for the minimum $(\sum,\min)$-value matroid partitioning problem in the same way. 
\end{proof}